\documentclass[pra,twocolumn,showpacs]{revtex4}

\usepackage{amsmath}
\usepackage{amsthm}

\theoremstyle{remark}

\begin{document}

\newtheorem{prop}{Proposition}

\title{Monogamy inequality for entanglement and 
local contextuality}
\author{S. Camalet}
\affiliation{Laboratoire de Physique Th\'eorique 
de la Mati\`ere Condens\'ee, UMR 7600, Sorbonne 
Universit\'es, UPMC Univ Paris 06, F-75005, 
Paris, France}

\begin{abstract}
We derive a monogamy inequality for entanglement 
and local contextuality, for any finite bipartite system. 
It essentially results from the relations between 
the entropy of a local state and the entanglement 
of the global state, and between the purity of a state, 
in the sense of majorization, and its ability to violate 
a given state-dependent noncontextuality inequality. 
We build an explicit entanglement monotone that 
satisfies the found monogamy inequality. 
An important consequence of this inequality, is that 
there are global states too entangled to violate 
the local noncontextuality inequality.
\end{abstract} 


\maketitle

\section{Introduction}  

One of the most important property of quantum 
entanglement, is known as entanglement 
monogamy \cite{HHHH,CKW}. Consider two 
systems, say A and B, in a maximally entangled 
state. Since this state is pure, there is no correlation 
between A and any third system, say C. 
In this extreme case, the entanglement 
between A and B, is maximum, and that 
between A and C, vanishes. In the general case, 
there is a trade-off between the two amounts 
of entanglement. Expressing it in quantitative terms, 
requires to specify a measure of entanglement 
\cite{HHHH,PV,V}. Monogamy inequalities have 
been derived, first for three qubits, in terms of 
squared concurrence \cite{CKW}, and then, 
for larger systems, and using different measures 
of entanglement 
\cite{OV,KW,BXW,SBYYC,LTSL,LDHPAW}. 
Related works consider nonlocality tests based 
on the Clauser-Horne-Shimony-Holt (CHSH) 
inequality \cite{B,CHSH}. 
When this inequality is violated for A and B, it is 
necessarily satisfied for A and C, if the same 
measurements are performed on A in both tests 
\cite{To,TV}. In contrast, monogamy inequalities 
for entanglement, do not involve specific 
observables, since the amount of entanglement 
between two systems, depends only on 
their common quantum state. 

It has been shown that 
the Klyachko-Can-Binicio\u glu-Shumovski (KCBS)
noncontextuality inequality \cite{KCBS}, for A only, 
cannot be violated toghether with the CHSH locality 
inequality \cite{KCK}, or the ${\cal I}_{3322}$ 
inequality \cite{SR}, when the same measurements 
are carried out on A in both tests. One can thus 
wonder whether there is a monogamy relation 
between entanglement and local contextuality. 
Such a relation must involve only the states of 
the global system, and of the considered local 
system, not particular observables. The ability 
of a state to disobey a noncontextuality inequality 
with few observables, is determined by 
its eigenvalues \cite{KK,XSC,PRA,RH}. 
Moreover, if a state is less pure, in the sense 
of majorization \cite{HLP,MO}, than a state that 
always satisfies a noncontextuality inequality, 
then it also cannot violate this inequality. 
The purity of the state of A, and the entanglement 
between A and B, clearly influence each other. 
To see it, consider the following two extreme cases. 
If A is in a pure state, A and B are uncorrelated. 
If A and B are maximally entangled, the reduced 
density operator for A is the maximally mixed state, 
which is majorized by any other one. 

In this paper, we derive a monogamy inequality 
for entanglement and local contextuality, for any 
finite bipartite system. To do so, we exploit 
the above mentioned relations between purity 
and contextuality, and between entanglement 
and local purity. We first show, in Sec. \ref{Iele}, 
that, for any entanglement monotone, 
the entanglement between A and B, cannot exceed 
a function of the state of A, that has the essential 
properties of an entropy \cite{BZHPL}. This result 
expresses quantitatively how the purity of a local 
state, and the entanglement of the global state, 
constrain each other. Then, in Sec. \ref{eni}, 
we define, from any given state-dependent 
noncontextuality inequality, involving dichotomic 
observables, an entropic measure, which dictates, 
for a specific size of A, wether its state can disobey 
the inequality. Finally, in Sec. \ref{Melc}, we build 
an explicit entanglement monotone, which is upper 
bounded by this particular entropy function. 
An important consequence of the found monogamy 
inequality, is that there are global states too 
entangled to violate the local noncontextuality 
inequality. For four-level systems and the CHSH 
inequality, we obtain a simple condition, in terms 
of a readily computable quantity, that determines 
such states, in Sec. \ref{Cem}.

\section{Relation between entanglement and local 
entropy}\label{Iele}

We consider a measure $E$ of the entanglement 
between any two finite systems. The value $E(\rho)$, 
where $\rho$ is the state of the global system, 
consisting of the local systems A and B, is positive, 
and vanishes if $\rho$ is not entangled. Moreover, 
it does not increase when two operators carry out 
local operations, and communicate classically. 
Such a function $E$ is an entanglement monotone. 
More specifically, $E[\Lambda(\rho)] \le E(\rho)$ for 
transformations $\Lambda$ composed of local 
operations $\rho \mapsto \sum_k 
M_k^{\phantom{\dag}}\rho M_k^{\dag} $, and maps 
of the form $\rho \mapsto \sum_k 
M_k^{\phantom{\dag}}\rho M_k^{\dag} 
\otimes | k \rangle \langle k |$, where $M_k$ acts on 
one local system only, 
$\sum_k M_k^{\dag}M_k^{\phantom{\dag}}$ is equal 
to the corresponding identity operator, 
and $| k \rangle$ are orthonormal states of 
an ancilla close to the other system \cite{HHHH}. 
We reiterate that $E$ is defined for local systems of 
any sizes. Clearly, the last map above transforms 
states of a system, into states of a different system. 
Moreover, $M_k$ can be a linear operator from 
the Hilbert space of a system, to that of one of 
its subsystems, or to that of a local larger system 
\cite{V}. 

We are interested in the constraint on the reduced 
density operator of a local system, set by 
the entanglement $E(\rho)$. To express it, we define, 
for states $\rho_\mathrm{A}$ of a $d$-level system A, 
\begin{equation}
S_d(\rho_\mathrm{A}) \equiv 
\max_{\rho \in {\cal C}(\rho_\mathrm{A})} E(\rho) ,
\label{S}
\end{equation}
where ${\cal C}(\rho_\mathrm{A})$ is the set of 
all states $\rho$ of all composite systems consisting 
of A, and another system, such that the reduced 
density operator for A is $\rho_\mathrm{A}$. 
For any system B, and any state $\rho$ of 
the global system AB, consisting of A and B, 
$$S_d(\operatorname{tr}_\mathrm{B} \rho) 
\ge E(\rho),$$ 
where $\operatorname{tr}_\mathrm{B}$ denotes 
the partial trace over B. As we will see, the equality 
is reached when $\rho$ is pure. We show below that 
the functions \eqref{S} have the essential properties 
of the familiar entropies (von Neumann, R\'enyi, 
Tsallis, \ldots) \cite{BZHPL}. Thus, the above 
inequality expresses how the purity of the local 
state $\operatorname{tr}_\mathrm{B} \rho$, and 
the entanglement of the global state $\rho$, 
constrain each other. We remark that this inequality 
is not necessarily satisfied if $S_d$ is replaced by 
an arbitrary entropy function. For distillable 
entanglement, entanglement cost, entanglement 
of formation, and relative entropy of entanglement, 
eq.\eqref{S} gives the von Neumann entropy 
\cite{HHHH,BBPS,VP}. For robustness and negativity, 
$S_d(\rho_\mathrm{A})$ is simply related to 
the 1/2-R\'enyi entropy \cite{VT,ViWe}.
\begin{prop} \label{Sp}
The functions \eqref{S} satisfy
\begin{equation}
S_d(\rho_\mathrm{A})=s({\boldsymbol p}) , 
\label{Ss}
\end{equation}
where $\boldsymbol p$ is the vector made up 
of the nonzero eigenvalues of $\rho_A$, 
in decreasing order, and $s$ does not depend 
on $d$, vanishes for ${\boldsymbol p}=1$, 
and obeys 
$s({\boldsymbol q}) \le s({\boldsymbol p})$ when 
${\boldsymbol q}$ majorizes ${\boldsymbol p}$. 
\end{prop}
\begin{proof}
Consider any system B', and any state 
$\rho \in {\cal C}(\rho_\mathrm{A})$ of the composite 
system AB'. Denote its eigenvalues by $\lambda_m$, 
and its eigenstates by $| \psi_m \rangle$. Let us 
introduce a third system, say B'',  which constitutes, 
together with B', system B. Provided the Hilbert 
space dimension of B'' is large enough, $\rho$ can be 
written as $\rho=\operatorname{tr}_\mathrm{B''}
| \Psi \rangle \langle \Psi |$, where 
$| \Psi \rangle=\sum_m 
\sqrt{\lambda_m} | \psi_m \rangle | \phi_m \rangle$ 
is a pure state of system AB, with orthonormal 
states $| \phi_m \rangle$ of B''. As 
$\operatorname{tr}_\mathrm{B''}$ is a local 
operation, on B, $E(\rho) \le s$ where 
$s=E(| \Psi \rangle \langle \Psi |)$. 

Since $\operatorname{tr}_\mathrm{B}
| \Psi \rangle \langle \Psi |=\sum_i p_i 
| i \rangle \langle i |$, where $p_i$ are 
the nonzero eigenvalues of $\rho_\mathrm{A}$, 
and $|i \rangle$ are the corresponding eigenstates, 
$| \Psi \rangle=\sum_i \sqrt{p_i} 
| i \rangle | \chi_i \rangle$, where 
$| \chi_i \rangle$ are orthonormal states 
of B. For any pure state $| \Psi' \rangle$ of AB, 
with Schmidt coefficients $\sqrt{p_i}$, there are 
unitary operators $U_\mathrm{A}$ and 
$U_\mathrm{B}$, acting on A and B, respectively, 
such that $| \Psi' \rangle=U_\mathrm{A}
\otimes U_\mathrm{B} | \Psi \rangle$. 
Thus, $| \Psi \rangle  \langle \Psi |$ and 
$| \Psi' \rangle \langle \Psi' |$ can be transformed 
into each other by local operations. Consequently, 
$E(| \Psi' \rangle \langle \Psi' |)=s$, and hence, 
$s$ is a function of $\boldsymbol p$ only. Since 
$\rho$ is an arbitrary state of 
${\cal C}(\rho_\mathrm{A})$, 
and $| \Psi \rangle \langle \Psi | 
\in {\cal C}(\rho_\mathrm{A})$, 
$S_d(\rho_\mathrm{A})=s$.

If $\boldsymbol p=1$, $| \Psi \rangle$ is a product 
state, and so $s=0$. 

Consider $| \Phi \rangle=\sum_i 
\sqrt{q_i} | i \rangle | \chi_i \rangle$ 
with $\boldsymbol q$ majorizing $\boldsymbol p$. 
We have $s({\boldsymbol p})
=E(| \Psi \rangle \langle \Psi |)$ and 
$s({\boldsymbol q})=E(| \Phi \rangle \langle \Phi |)$. 
Since $| \Psi \rangle \langle \Psi |$ can be changed 
into $| \Phi \rangle \langle \Phi |$ by local operations 
and classical communication \cite{N}, 
$s({\boldsymbol p}) \ge s({\boldsymbol q})$. 
\end{proof}

Relation \eqref{Ss} means not only that
\begin{equation}
S_d(U \rho_\mathrm{A} U^\dag) 
= S_d(\rho_\mathrm{A}) , \label{U}
\end{equation}
where $U$ is any unitary operator of A, 
but also that
\begin{equation}
S_{d+1} \left( \sum_{i=1}^{d} p_i 
|\tilde \imath \rangle \langle \tilde \imath | \right) 
= S_d \left( \sum_{i=1}^{d} p_i 
| i \rangle \langle i | \right) \label{dim} ,
\end{equation}
where $\{ | i \rangle \}_{i=1}^{d}$ and 
$\{ | \tilde \imath \rangle \}_{i=1}^{d+1}$ are 
orthonormal bases of the considered Hilbert spaces, 
and the probabilities 
$p_i$ obey $\sum_{i=1}^{d} p_i=1$. The classical 
form of equation \eqref{dim} is known as 
the expansibility property, and is an essential 
requirement for an entropic measure \cite{BZHPL,K}.

\section{Entropies from  noncontextuality inequalities}
\label{eni}

Our aim is to study the influence of the entanglement 
between systems A and B, on contextuality tests 
involving only A. This local contextuality can be 
revealed by considering $N$ dichotomic observables 
$A_k$ of A, such that each observable is compatible 
with some other ones, but not with all. We restrict 
ourselves to the usual case of projective 
measurements with two outcomes. When evaluated 
with a noncontextual hidden variable theory, 
the correlations of the compatible observables, 
satisfy inequalities, which can be violated 
by quantum states. Such a noncontextuality 
inequality reads
\begin{equation}
\sum_n x_n \big\langle \prod_{k \in {\cal E}_n}  
A_k \big\rangle \le 1 , \label{nci}
\end{equation}
where ${\cal E}_n$ are subsets of 
$\{ 1, \ldots , N \}$, of any possible size, 
and $\langle \ldots \rangle
=\operatorname{tr}(\rho_\mathrm{A} \ldots)$ 
is the average with respect to the density matrix 
$\rho_\mathrm{A}$. The observables $A_k$ 
and $A_l$ commute with each other when 
$k,l \in {\cal E}_n$. The coefficients $x_n$ 
are such that the maximum value of the left-hand 
side of eq.\eqref{nci}, is $1$ for noncontextual 
hidden-variable models, i.e., there are 
$a_k=\pm 1$, such that 
$\sum_n x_n \prod_{k \in {\cal E}_n}  a_k = 1$. 
The familiar CHSH and KCBS inequalities 
\cite{B,CHSH,KCBS}, for example, can be cast 
into the form \eqref{nci}. Let us define
\begin{equation}
C_d(\rho_\mathrm{A}) 
\equiv \sup_{{\bf A} \in {\cal A}_d} 
\operatorname{tr} \left( \rho_\mathrm{A} 
\sum_n x_n \prod_{k \in {\cal E}_n}  A_k \right) , 
\nonumber 
\end{equation}
where $d$ is the Hilbert space dimension of A, 
${\bf A}$ stands for $( A_1, \ldots , A_N )$, 
and ${\cal A}_d$ is the set of all ${\bf A}$ consisting 
of dichotomic observables $A_k$, such that 
$[A_k,A_l]=0$ for $k,l \in {\cal E}_n$. 
By construction, for a state $\rho_\mathrm{A}$ 
such that $C_d(\rho_\mathrm{A}) > 1$, 
there are observables $A_k$ with which 
inequality \eqref{nci} is violated. 

It has been shown that 
$C_d(\rho_\mathrm{A})
=c_d(\boldsymbol p)$, where $\boldsymbol p$ 
is the vector made up of the eigenvalues 
of $\rho_\mathrm{A}$, in decreasing order, 
and $c_d$ satisfies 
$c_d({\boldsymbol q}) \ge c_d({\boldsymbol p})$ 
when ${\boldsymbol q}$ majorizes 
${\boldsymbol p}$ \cite{PRA}. 
However, the functions $C_d$ do not obey 
the expansibility condition \eqref{dim}, and really 
depend on the dimension $d$. Due to 
the above-mentioned property of $c_d$, $C_d$ 
reaches its maximum, $C_d^{max} \equiv c_d(1)$, 
for pure states. We assume that there 
are dimensions $d$ for which $C_d^{max}>1$. 
For these values of $d$, inequality \eqref{nci} 
constitutes a proper contextuality test, since it is 
not always satisfied. Note that, for some 
state-dependent noncontextuality inequalities, 
$C_d^{max}$ does not depend on $d$, 
provided it is larger than some value  
\cite{AQBCC}. For CHSH inequality, for example, 
it is equal to $\sqrt{2}$, for $d \ge 4$. 
We also remark that the operators $A_k=a_k A$, 
where $A$ is any dichotomic observable, and 
$a_k=\pm 1$ are such that  
$\sum_n x_n \prod_{k \in {\cal E}_n}  a_k = 1$, 
obviously fulfill the above-stated commutation 
relations. Such a case describes a set-up 
that consists of $N$ measurement apparatuses 
corresponding to the same observable $A$. 
As a consequence, if all products 
in eq.\eqref{nci}, have an even number 
of terms, $C_d \ge 1$.  

To study the impact of the entanglement between 
A and B, on the local contextuality test \eqref{nci}, 
we define 
\begin{equation}
S^{lc}_{d}(\rho_\mathrm{A})\equiv C_{d_0}^{max} 
- \max_{\{ | \tilde \imath \rangle \} } 
tC_{d_0}\left( \sum_{i,j=1}^{d_1} 
\langle \tilde \imath | \rho_\mathrm{A} 
| \tilde \jmath \rangle 
| i \rangle \langle j | /t \right) , \label{gE}
\end{equation}
where $d_0$ is a specific dimension, 
$d_1=\min \{ d,d_0 \}$, $t=\sum_{i=1}^{d_1} 
\langle \tilde \imath | \rho_\mathrm{A} | 
\tilde \imath \rangle$, $\{ | i \rangle \}_{i=1}^{d_0}$ 
is an orthonormal basis, and the maximum is taken 
over the orthonormal bases 
$\{ | \tilde \imath \rangle \}_{i=1}^{d}$ of A. 
Since $\mathrm{C}_{d_0}$ obeys eq.\eqref{U}, 
the definition \eqref{gE} does not depend on 
any particular basis. For $d=d_0$, it reduces to 
$S^{lc}_{d_0}(\rho_\mathrm{A})=C_{d_0}^{max}
-C_{d_0}(\rho_\mathrm{A})$, but, for $d \ne d_0$, 
$S^{lc}_{d}$ and $C_{d}$ are not simply related to 
each other. As $C_{d_0}^{max}$ is the maximum 
value of $C_{d_0}$, the functions \eqref{gE} are 
positive. As a consequence of the result below, they 
also fulfill the properties enumerated in proposition 
\ref{Sp}. Note that there are state-independent 
noncontextuality inequalities \cite{P,M,C,YO} for which 
the definition \eqref{gE} gives zero for any state, 
and is thus of no use. In this case, no meaningful 
entanglement monotone $E$ can obey eq.\eqref{S} 
with $S^{lc}_d$, since the only possibility is $E=0$. 
In the following, we use the notation
 $\boldsymbol \lambda (M)$ for the vector made up 
of the eigenvalues of the Hermitian operator $M$, 
in decreasing order. 

\begin{prop} The functions \eqref{gE} satisfy 
\begin{equation}
S^{lc}_{d}(\rho_\mathrm{A})=C_{d_0}^{max}-
\sup_{{\boldsymbol \mu} \in \Lambda} 
\left( \sum_{i=1}^{d_0} \mu_i p_i \right) , \nonumber
\end{equation} 
where $\Lambda$ is the set of 
all vectors ${\boldsymbol 
\lambda}(\sum_n x_n \prod_{k \in {\cal E}_n}  A_k)$, 
with $(A_1, \ldots, A_N) \in {\cal A}_{d_0}$, 
$p_i=\lambda_i (\rho_\mathrm{A})$ for $i \le d$, 
and $p_i=0$ for $i > d$.
\end{prop}
\begin{proof}
Consider any orthonormal 
bases $\{ | \tilde \imath \rangle \}_{i=1}^{d}$
and $\{ | i \rangle \}_{i=1}^{d_0}$, 
and define $\Omega=\sum_{i,j=1}^{d_1} 
\langle  \tilde \imath | \rho_\mathrm{A} 
| \tilde \jmath \rangle | i \rangle \langle j |$, 
where $d_1=\min \{ d,d_0 \}$, and 
the state $\omega=t^{-1} \Omega$, where 
$t=\operatorname{tr} \Omega$. It has been 
shown that $C_{d_0}(\omega) = 
\sup_{{\boldsymbol \mu} \in \Lambda} 
[{\boldsymbol \mu} \cdot 
{\boldsymbol \lambda} (\omega) ]$, 
where 
${\bf a} \cdot {\bf b}=\sum_{i=1}^{d_0} a_i b_i$ 
\cite{PRA}. Since $t{\boldsymbol \lambda} 
(\omega)={\boldsymbol \lambda} (\Omega)$, 
$tC_{d_0}(\omega) = 
\sup_{{\boldsymbol \mu} \in \Lambda} 
[{\boldsymbol \mu} 
\cdot {\boldsymbol \lambda} (\Omega) ]$. 

We denote 
${\boldsymbol \lambda}(\rho_\mathrm{A})$ 
by $\bf p$. For $d > d_0$, the matrix representation 
of $\Omega$, in the basis $\{ | i \rangle \}$, is 
a diagonal block of that of $\rho_\mathrm{A}$, 
in the basis $\{ | \tilde \imath \rangle \}$. 
Thus, $\bf p$ weakly submajorizes 
${\boldsymbol \lambda}(\Omega)$ \cite{MO}, 
and so, for $j=1, \ldots, d_0$, 
$R_j \equiv \sum_{i=1}^j [\lambda_i(\Omega)-p_i]$ 
is negative. Consequently, 
for any ${\boldsymbol \mu} \in \Lambda$, 
${\boldsymbol \mu} \cdot 
[{\boldsymbol \lambda}(\Omega)  - {\bf p}^{[d_0]}]
=\sum_{j=1}^{d_0-1}(\mu_j-\mu_{j+1})R_j 
+ \mu_{d_0}R_{d_0}  \le 0$, where ${\bf p}^{[d_0]}$ 
is made up of the $d_0$ largest $p_i$, 
in decreasing order. 
Hence, $tC_{d_0}(\omega) \le 
\sup_{{\boldsymbol \mu} \in \Lambda} 
({\boldsymbol \mu} \cdot {\bf p}^{[d_0]})$. 
For $d \le d_0$, this inequality becomes an equality, 
with ${\bf p}^{[d_0]}$ made up of the $p_i$, 
in decreasing order, followed by $d_0-d$ zeros, 
since 
${\boldsymbol \lambda}(\Omega)={\bf p}^{[d_0]}$.

For any $d$, 
when $\{ | \tilde \imath \rangle \}_{i=1}^{d}$ 
is such that $\rho_\mathrm{A}=\sum_{i=1}^d p_i 
| \tilde \imath \rangle \langle \tilde \imath |$, 
${\boldsymbol \lambda}(\Omega)={\bf p}^{[d_0]}$, 
which finishes the proof.
\end{proof}

\section{Monogamy of entanglement and local 
contextuality}\label{Melc}

The functions \eqref{gE} have all the required 
characteristics to satisfy eq.\eqref{S} with 
an entanglement monotone $E$. It remains to show 
that there is indeed such a measure $E$. 
This can be achieved, by using the convex 
roof method \cite{HHHH}, since, due to the convexity 
of $C_{d_0}$ \cite{PRA}, $S^{lc}_d$, given 
by eq.\eqref{gE}, is concave.
\begin{prop}\label{Ep}
Consider, for any composite system AB, and any state 
$\rho$ of AB,  
\begin{equation}
E^{cr}(\rho)\equiv \inf_{\{P_m,| \Psi_m \rangle \} 
\in {\cal D}(\rho)} \sum_m P_m {S}_d 
\left(\operatorname{tr}_\mathrm{B} 
| \Psi_m \rangle \langle \Psi_m | \right) , \label{E}
\end{equation} 
where ${\cal D}(\rho)$ is the set 
of all ensembles $\{P_m,| \Psi_m \rangle \}$ 
such that 
$\sum_m P_m 
| \Psi_m \rangle \langle \Psi_m |=\rho$, 
$d$ is the Hilbert space dimension of A, 
and ${S}_d$ are positive concave functions obeying 
eq.\eqref{Ss}, and vanishing for pure states.

The function $E^{cr}$ is an entanglement 
monotone, and satisfies eq.\eqref{S} with $S_d$.
\end{prop}
\begin{proof}
We first consider that $\rho$ is not entangled. 
Then, by definition, $\rho$  is a mixture of pure 
product states $| \Psi_m \rangle$. The corresponding 
states $\operatorname{tr}_\mathrm{B} 
| \Psi_m \rangle \langle \Psi_m |$ are pure, 
and hence $E^{cr}(\rho)=0$. 

Let us now prove that interchanging A and B 
does not modify expression \eqref{E}. The reduced 
density operators $\operatorname{tr}_\mathrm{B} 
| \Psi_m \rangle \langle \Psi_m |$ 
and $\operatorname{tr}_\mathrm{A} 
| \Psi_m \rangle \langle \Psi_m |$, have the same 
nonvanishing eigenvalues. Thus, due to eq.\eqref{Ss}, 
${S}_d (\operatorname{tr}_\mathrm{B} 
| \Psi_m \rangle \langle \Psi_m | )$ 
in eq.\eqref{E}, can be replaced by ${S}_{d'} 
(\operatorname{tr}_\mathrm{A} 
| \Psi_m \rangle \langle \Psi_m | )$, where 
$d'$ is the Hilbert space dimension of B.

It follows from eq.\eqref{E} that $E^{cr}$ is convex 
\cite{V}. For operators $B_k$ of system B, such that 
$\sum_k B_k^{\dag} B_k^{\phantom{\dag}}$ is equal 
to its identity operator, the concavity of $S_d$ leads 
to $E^{cr}(\rho) \ge \sum_k p_k E^{cr}(\rho_k)$, 
where $p_k=\operatorname{tr} 
(B_k^{\dag} B_k^{\phantom{\dag}}\rho)$ 
and $\rho_k=
B_k^{\phantom{\dag}}\rho B_k^{\dag}/p_k$ \cite{V}. 
This inequality and the convexity of $E^{cr}$ 
ensure that $E^{cr}$ does not increase under local 
operations on B. With expression \eqref{E} rewritten 
as explained above, the same proof shows that this 
is also the case for local operations on A. 
Since $\rho_k$ and 
$\tilde \rho_k = \rho_k\otimes | k \rangle \langle k |$, 
where $| k \rangle$ is a pure state of an ancilla 
close to A, can be transformed into each other 
by local operations, 
$E^{cr}(\tilde \rho_k)=E^{cr}(\rho_k)$. Thus, 
$E^{cr}(\sum_k p_k \tilde \rho_k) \le E^{cr}(\rho)$, 
which finishes the proof that $E^{cr}$ is 
an entanglement monotone. 

Consider a given state $\rho_\mathrm{A}$ of A, and 
any state $\rho \in {\cal C}(\rho_\mathrm{A})$. 
The definition \eqref{E} and the concavity of $S_d$ 
give $E^{cr}(\rho) \le S_d(\rho_\mathrm{A})$. 
If $d' \ge d$, there are pure states $|\Psi \rangle$ 
of AB such that $\operatorname{tr}_\mathrm{B} 
|\Psi \rangle\langle \Psi | =\rho_\mathrm{A}$, 
and hence $E^{cr}(|\Psi \rangle\langle \Psi |)
=S_d(\rho_\mathrm{A})$. Consequently, 
$\max_{\rho \in {\cal C}(\rho_\mathrm{A})} 
E^{cr}(\rho)=S_d(\rho_\mathrm{A})$.
\end{proof}

We have thus, for a $d_0$-level system A, 
the monogamy inequality
\begin{equation}
E(\rho)+C_{d_0}(\rho_\mathrm{A}) 
\le C_{d_0}^{max} , \label{tdi}
\end{equation}
where $E$ is given by eq.\eqref{E} with 
the functions \eqref{gE}. Thus, 
the entanglement of A with B, as quantified 
by $E(\rho)$, restricts the value of the left side 
of inequality \eqref{nci}. In particular, for 
a state $\rho$ such that 
$E(\rho) \ge C_{d_0}^{max}-1$, 
this noncontextuality inequality cannot be violated. 
Equation \eqref{tdi} can also be read as an upper 
bound on the entanglement $E(\rho)$. 
In the extreme case of maximal violation of 
eq.\eqref{nci}, i.e., 
$C_{d_0}(\rho_\mathrm{A})=C_{d_0}^{max}$, 
it gives $E(\rho)=0$. 
There may be other entanglement monotones that 
coincide with the functions \eqref{gE} when $\rho$ 
is pure, and so satisfy inequality \eqref{tdi}. 
But, there is no entanglement monotone, for 
which eq.\eqref{tdi} is always an equality, since 
$C_{d_0}[\Lambda(\rho_\mathrm{A})]
\le C_{d_0}(\rho_\mathrm{A})$ for some local 
operations $\Lambda$ on A. Some noncontextuality 
inequalities \eqref{nci} are violated 
for any state $\rho$, which, in this case, necessarily 
satisfies $E(\rho) < C_{d_0}^{max}-1$. 
If the corresponding function $C_{d_0}$ is constant, 
$E=0$, and eq.\eqref{tdi} is trivially obeyed, and 
of no relevance. This is not surprising, since 
such a state-independent noncontextuality inequality 
is always maximally violated \cite{P,M,C,YO}. 
If $C_{d_0}$ is larger than unity, but not constant, 
eq.\eqref{tdi} still gives an upperbound, that depends 
on the entanglement beween A and B, for the left side 
of eq.\eqref{nci}.
 
\section{Computable measures of entanglement}
\label{Cem}

The monogamy inequality \eqref{tdi} involves 
an unusual entanglement monotone, defined 
from the considered noncontextuality inequality. 
Moreover, even familiar entanglement monotones 
are difficult to evaluate for an arbitrary density 
matrix $\rho$ \cite{H}. An exception is 
the negativity $(\| \rho^\Gamma \|-1)/2$, 
where $\| M \|=\operatorname{tr} \sqrt{MM^\dag}$ 
denotes the trace norm of operator $M$, and 
$\rho^\Gamma$ is a partial transpose of 
$\rho$ \cite{HHHH,ZHSL,ViWe}. 
There are entangled states with vanishing 
negativity. Other quantities can be used to detect 
entanglement, e.g., $\| {\cal R}(\rho) \|$, where 
${\cal R}$ is a matrix realignment map, which is 
not greater than 1 when $\rho$ is not entangled 
\cite{R1,R2,CW}. In ref.\cite{CAF}, a lower bound 
is derived for the entanglement of formation, 
in terms of
\begin{equation}
x\equiv \max \{ \| \rho^\Gamma \| 
, \| {\cal R}(\rho)\| \} , \label{x}
\end{equation}
which is readily computable. We show below that 
a similar bound can be obtained for 
any entanglement monotone of the form \eqref{E}. 
\begin{prop}
Consider an entanglement monotone $E^{cr}$ given 
by eq.\eqref{E}, two systems, A and B, of Hilbert 
space dimensions $d$ and $d'$, respectively, 
and the function $f$ defined, 
for $y\in [1,d^*]$, where $d^*=\min\{d,d'\}$, by 
\begin{equation}
f(y)\equiv \inf_{{\bf p} \in {\cal F}(y)}  s({\bf p}) , \label{f}
\end{equation}
where $s$ is given by eq.\eqref{Ss}, and ${\cal F}(y)$ 
is the set of the $d^*$-component probability vectors 
${\bf p}$, such that 
$(\sum_{i=1}^{d^*} \sqrt{p_i})^2=y$.

For any state $\rho$ of AB, $E^{cr}(\rho) \ge co(f)(x)$, 
where $co(f)$ is the convex hull of $f$, and $x$ is 
given by eq.\eqref{x}. 
\end{prop} 
\begin{proof}
Let us first show that $co(f)$ exists and is 
nondecreasing. Since $s$ is positive, $f \ge 0$, 
and thus, $f$ has a convex hull \cite{HUL}. It is 
the maximum of the convex functions not larger 
than $f$. As $f \ge 0$, $co(f)$ is positive.
The only element of ${\cal F}(1)$ is ${\bf p}=1$. 
Thus, $f(1)=0$, and hence, $co(f)(1)=0$. 
Consider $y_1$ and $y_2$ such that 
$1 \le y_1 \le y_2 \le d^*$. We have 
$y_1=\tau + (1-\tau)y_2$ 
with $\tau \in [0,1]$. So, using the convexity 
and positivity of $co(f)$, and $co(f)(1)=0$, 
we get $co(f)(y_1) \le co(f)(y_2)$.

Consider any ensemble 
$\{ P_m, | \Psi_m \rangle \} \in {\cal D}(\rho)$, 
and denote by ${\bf p}^{(m)}$ the $d^*$-component 
vector made up of the squared Schmidt coefficients 
of $| \Psi_m \rangle$, in decreasing order, 
possibly completed with zeros. By definition of $f$, 
$\sum_m P_m 
S_d(\operatorname{tr}_\mathrm{B} \rho_m)
\ge \sum_m P_m f (y_m) $, where 
$\rho_m= | \Psi_m \rangle\langle \Psi_m |$,
and $y_m=[\sum_{i=1}^{d^*} (p_i^{(m)})^{1/2}]^2$. 
The right side of this inequality is not smaller than 
$co(f)(y)$ where $y=\sum_m P_m y_m$. 
Since $y_m=\| \rho_m^\Gamma \|
=\|{\cal R} (\rho_m) \|$ \cite{R1,CAF2,ViWe}, 
and the trace norm is convex, $x \le y$. Using 
this inequality, and the monotonicity of $co(f)$, 
leads to the result.
\end{proof}

The above proposition, and the monogamy 
inequality \eqref{tdi}, give, for a $d_0$-level 
system A,
\begin{equation}
C_{d_0}(\rho_\mathrm{A}) 
\le C_{d_0}^{max}-co(f)(x) , 
\label{cof}
\end{equation}
where $f$ is given by eq.\eqref{f} with 
the functions \eqref{gE}, and $d^*=\min\{d_0,d'\}$ 
with $d'$ the Hilbert space dimension of B. 
For states $\rho$ such that 
$co(f)(x) \ge C_{d_0}^{max}-1$, 
eq.\eqref{nci} cannot be violated. As noted above, 
if all products in eq.\eqref{nci}, have an even 
number of terms, $C_d \ge 1$, and so, $co(f)$ 
can reach $C_{d_0}^{max}-1$ only for $x=d^*$, 
i.e., for maximally entangled states. 
However, even in this case, eq.\eqref{cof} can be 
useful to determine states too entangled to violate 
inequality \eqref{nci}. As an example, consider 
$d^*=d_0=4$, and the CHSH inequality, for which, 
as shown below, $C_4(\rho_\mathrm{A})
=\max \{ 1, C'_4(\rho_\mathrm{A}) \}$, where
\begin{equation}
C'_4(\rho_\mathrm{A})= 
\sqrt{2\left[ (p_1-p_4)^2+(p_2-p_3)^2 \right]} , 
\label{CHSH}
\end{equation}
with $p_i=\lambda_i(\rho_\mathrm{A})$. 
The functions \eqref{gE}, defined with $C'_4$, 
have all the necessary properties to obey 
eq.\eqref{cof} with the corresponding $f$. 
Using the method of Lagrange multipliers, 
we find 
$$f(y)=\sqrt{2}-(3\sqrt{y}+\sqrt{32-7y})^{3/2} 
(\sqrt{32-7y}-\sqrt{y})^{1/2}/32.$$ This function 
is convex, and is hence equal to its convex hull. 
It increases from 0 to $C_4^{max}=\sqrt{2}$. 
Consequently, for states $\rho$ such that 
$x \ge 2.95$, $C'_4(\rho_\mathrm{A}) \le 1$, 
and thus, the local CHSH inequality is always 
satisfied. 
\begin{proof}
For CHSH inequality, $C_d(\rho_\mathrm{A})
= \sup_{{\bf A} \in {\cal A}_d} \langle T \rangle/2$, 
where $\langle T \rangle
=\operatorname{tr}(\rho_\mathrm{A} T)$, 
and $T=A_1(A_2+A_4)+A_3(A_2-A_4)$. 
For $A_k=A_1$, $\langle T \rangle=2$, 
and hence $C_d \ge 1$. We are thus 
interested only in observables $A_k$ such 
that $\langle T \rangle$ can be larger than 2. 
For $d=4$, $A_k$ can be written as 
$A_k=\eta_k(2\Pi_k-I)$ where $\eta_k=\pm 1$, 
$I$ is the identity operator, and $\Pi_k$ is 
a projector of rank not greater than 2. Using 
this expression, one finds $T^2= 4 I \pm 16 R$, 
where $R=[\Pi_1,\Pi_3] [\Pi_2,\Pi_4]$ 
\cite{L}. If $R=0$, the eigenvalues of $T$ can 
only be $2$ and $-2$, and so 
$\langle T \rangle \le 2$. 
We thus search for the projectors $\Pi_k$ for 
which $R\ne 0$. If two commuting 
$\Pi_k$ and $\Pi_l$, obey $\Pi_k\Pi_l=0$ 
or $\Pi_k\Pi_l=\Pi_k$, then $R=0$. Consequently, 
the sought projectors are of rank 2, and 
such that, for $[\Pi_k,\Pi_l]=0$, $\Pi_k\Pi_l$ is 
a rank-$1$ projector. This gives 
$\Pi_k=|k \rangle \langle k |
+|k' \rangle \langle k' |$, with
\begin{multline}
|1 \rangle=| \tilde 1 \rangle,
|1' \rangle=|2 \rangle=| \tilde 2 \rangle , 
|2' \rangle=| \tilde 3 \rangle , 
|3 \rangle=\nu_1 | \tilde 2 \rangle 
+ \hat \nu_1 | \tilde 3 \rangle ,  \\
|3' \rangle=\nu_1 | \tilde 1 \rangle 
+ \hat \nu_1  | \tilde 4 \rangle , 
|4 \rangle=\nu_2 | \tilde 1 \rangle 
+ \hat \nu_2 | \tilde 2 \rangle , 
|4' \rangle= \hat \nu_2 | \tilde 3 \rangle 
+ \nu_2 | \tilde 4 \rangle , 
\nonumber
\end{multline}
where $\{ | \tilde \imath \rangle \}_{i=1}^4$ 
is any orthonormal basis, 
$|\hat \nu_k|^2+\nu_k^2=1$, and 
$\nu_k \in [0,1]$. For these projectors, 
the eigenvalues of $[\Pi_k,\Pi_{k+2}]$, 
where $k=1$ or 2, are $\pm i \nu_k|\hat \nu_k|$. 
Since these two commutators commute 
with each other, and $\operatorname{tr}R=0$, 
there are $A_k$ such that 
${\boldsymbol \lambda}(R)=4r(1,1, -1,-1)$, 
where $r \in ]0,1]$. For these $A_k$, 
$\mathrm{tr}(A_kA_l)=0$ for commuting 
$A_k$ and $A_l$, and hence 
$\operatorname{tr}T=0$. 
So, ${\boldsymbol \lambda}(T)
=(r_+,r_-, -r_-,-r_+)$, where $r_\pm=2\sqrt{1\pm r}$. 
Maximising 
${\boldsymbol \lambda}(T) \cdot {\boldsymbol p}$ 
over $r$, leads to eq.\eqref{CHSH}.
\end{proof} 

\section{Conclusion}

In summary, a monogamy inequality for entanglement 
and local contextuality, has been derived. It involves 
an entanglement monotone that depends on 
the considered noncontextuality inequality, and 
the Hilbert space dimension of the local system. 
It essentially results from the relations between 
the entanglement of the global state and the entropy 
of the local state, and between the eigenvalues 
of the local state and its ability to disobey 
the noncontextuality inequality. 
Thus, other entanglement monotones, 
different from the one we have built, may satisfy 
the same monogamy inequality. A consequence 
of the found monogamy, is that there are global 
states so entangled that they cannot violate 
the noncontextuality inequality. 
The obtained monogamy inequality relates 
entanglement {\it per se} to local contextuality. 
It would thus be of interest to find out if there 
are global states, that are Bell-local \cite{B}, 
but still entangled enough to prevent 
the violation of a local noncontextuality 
inequality.

\end{document}